\documentclass[11pt,a4paper]{article}

\usepackage{titlesec}
\usepackage{amsmath,amssymb,amsthm,mathtools}
\usepackage{array}
\usepackage{booktabs}
\usepackage{enumitem}
\usepackage{flafter}
\usepackage{placeins}
\usepackage{xcolor}
\usepackage{microtype}
\usepackage{bm}
\usepackage{graphicx}
\usepackage{url}
\usepackage[pdftex,pagebackref,colorlinks]{hyperref}
\hypersetup{
  pdftitle={Linear programming bounds for binary and ternary LCD Codes},
  pdfauthor={Ming-Hsuan Kang; Maosheng Xiong}
}
\date{}

\newcommand{\mb}[1]{\mathbf{#1}}
\newcommand{\F}{\mathbb F}
\newcommand{\Z}{\mathbb Z}

\newcommand{\wt}{\operatorname{wt}}
\newcommand{\abs}[1]{\left|#1\right|}
\newcommand{\set}[1]{\left\{#1\right\}}

\newcommand{\ii}{\mathrm{i}}
\newcommand{\Real}{\operatorname{Re}}

\theoremstyle{plain}
\newtheorem{theorem}{Theorem}[section]
\newtheorem{proposition}[theorem]{Proposition}
\newtheorem{lemma}[theorem]{Lemma}

\theoremstyle{definition}
\newtheorem{definition}[theorem]{Definition}

\theoremstyle{remark}
\newtheorem{remark}[theorem]{Remark}

\begin{document}

\title{Linear programming bounds for binary and ternary LCD Codes}

% Author affiliations are placed in a title footnote, as in LX-CPP.tex.
\author{Ming-Hsuan Kang\thanks{Corresponding author.
\newline\indent Ming-Hsuan Kang, Department of Applied Mathematics,
National Yang Ming Chiao Tung University, Hsinchu, Taiwan
(e-mail: \texttt{kmsming@gmail.com}).
\newline\indent Maosheng Xiong, Department of Mathematics,
The Hong Kong University of Science and Technology, Clear Water Bay,
Hong Kong, P. R. China (e-mail: \texttt{mamsxiong@ust.hk}).},\, Maosheng Xiong}

\maketitle

\begin{abstract}
We derive linear programming (LP) bounds on the minimum distance of binary and ternary linear complementary dual (LCD) codes by imposing arithmetic constraints on their weight enumerators. Special values of the weight enumerator give finitely many Gauss phases, each of which yields linear equations in the ordinary weight-distribution variables. The resulting bounds strengthen the real-valued LCD counting LP without introducing additional variables; both the number of branches and the number of added equations per branch are bounded independently of the code length. Exact certificates establish strict improvements for $62$ binary parameter pairs of length at most $20$ and $39$ ternary pairs of length at most $14$. For four binary pairs, the bounds also improve the joint-weight-enumerator LP while using fewer variables per branch. These comparisons show that Gauss-phase information provides a compact strengthening of existing LP relaxations for LCD codes.
\end{abstract}



\section{Introduction}
\label{sec:introduction}
%\input intro.tex
% LP-first Introduction for review; uses the preamble in LCD-KX-1.tex.
% The unnumbered display recalls the LCD criteria from thm:ternary-gauss and thm:binary-gauss.
% The closing roadmap follows the agreed revised organization.

Let $\F_q$ be the finite field of order $q$.  A linear code
$C\leq\F_q^n$ is called \emph{linear complementary dual} (LCD) if
$C\cap C^\perp=\set{\bm 0}$, where $C^\perp$ denotes its Euclidean dual.
Introduced by Massey \cite{Massey1992}, LCD codes form an important
family in coding theory, with applications in multisecret sharing
\cite{AlahmadiEtAl2020}, lattice constructions \cite{HouOggier2016},
and protection against side-channel and fault-injection attacks
\cite{BringerEtAl2014,CarletGuilley2016}, as well as connections with
complements in projective space motivated by network coding
\cite{BraunEtzionVardy2013}.
These applications have stimulated extensive research on LCD codes,
including their construction \cite{Jin2017,LiDingLi2017}, algebraic
structure \cite{CarletEtAl2019}, and classification and optimal
parameters \cite{Bouyuklieva2021,HaradaSaito2019}.
For recent developments, interested readers may refer to
\cite{MondalLee2026,XieZhu2025,ZhangZhengWang2026} and the references
therein.

We seek upper bounds on the minimum distance of LCD codes with fixed
length $n$ and dimension $k$.  For $q>3$, every linear code over $\F_q$
is monomially equivalent to an LCD code \cite{CarletEtAl2018}, so the
LCD condition imposes no further restriction on the largest attainable
minimum distance.  Therefore we focus only on binary and ternary cases.

Delsarte's linear programming (LP) method~\cite{Delsarte1973} is
one of the most effective tools for bounding codes over fixed
alphabets; see, for example, the McEliece--Rodemich--Rumsey--Welch
(MRRW) rate bounds~\cite{McElieceEtAl1977}.  For linear codes, the LP
incorporates the MacWilliams identities~\cite{MacWilliams1963} as
linear constraints relating the weight distributions of a code and
its dual.  The framework also accommodates additional constraints
reflecting the structure of particular code
families, as illustrated by recent LP bounds for locally recoverable
codes~\cite{gruicaLRCSDuality2026,LWX-LRC}.  

For an LCD code $C$ and any $t>0$, the weight-$t$ codewords in
$C$ and $C^\perp$ are distinct, so their combined number cannot exceed
the number of weight-$t$ vectors in the ambient space.  Dougherty et al.\
\cite{DoughertyEtAl2017} combined these counting inequalities with the
MacWilliams identities to develop an LP bound for binary LCD codes
using only the ordinary weight distributions.  Karabakla and
\"Ozkaya \cite{KarabaklaOzkaya2026} extended this method
to ternary codes.  The reported numerical results show that this
counting approach yields useful bounds in certain parameter ranges.
We call the real-valued formulation at fixed dimension the
\emph{LCD counting LP}; its constraints are given in
Definition~\ref{def:lcd-counting-lp}.

Alahmadi et al.\ \cite{AlahmadiEtAl2019} developed the
\emph{joint-weight-enumerator LP} for binary LCD codes using the
joint weight enumerator of a code and its dual.  This enumerator
records coordinate patterns of codeword pairs,
allowing the orthogonality between the codes and the LCD
condition to be expressed as linear constraints on its coefficients.
Their computations for lengths $n\leq16$ improve earlier LP
bounds for many parameter pairs~\cite{DoughertyEtAl2017}.  However, the $O(n^3)$ joint
coefficients make the LP computationally demanding.
A direct ternary extension is conceivable, but would involve more
coordinate patterns and likely even greater computational cost.
Definition~\ref{def:mixed-benchmark} in
Appendix~\ref{sec:mixed-comparison} specifies this formulation used
for our comparisons.

%\noindent\textbf{Our contribution.}

In this paper, we develop new LP bounds for binary and ternary LCD
codes that strengthen the LCD counting LP while retaining its
$O(n)$ ordinary weight variables.  We use a finite collection of LPs,
each imposing only a bounded number of additional linear constraints
on these variables.  The number of LPs and the number of additional constraints in each are bounded independently of the code length. Computations show strict improvements over the
real-valued LCD counting LP for both
alphabets, as well as four strict improvements over the
joint-weight-enumerator LP~\cite{AlahmadiEtAl2019} on selected binary
parameter pairs.  These four comparisons are verified exactly using
rational feasibility witnesses and integer Farkas certificates.

Our LP bounds use the following characterizations of binary and
ternary LCD codes through special values of their weight
enumerators.  For a $k$-dimensional linear code $C\leq\F_q^n$, write
\[
 W_C(z)=\sum_{\bm u\in C}z^{\wt(\bm u)},
\]
where $\wt(\bm u)$ is the Hamming weight of the vector $\bm u$.  Set $\ii=\sqrt{-1}$ and
$\omega=e^{2\pi\ii/3}$.  For $q\in\set{2,3}$, one has
\[
 C\text{ is LCD}
 \quad\Longleftrightarrow\quad
 \begin{cases}
 \abs{W_C(\ii)}=2^{k/2},&\text{if }q=2,\\
 \abs{W_C(\omega)}=3^{k/2},&\text{if }q=3.
 \end{cases}
\]

While preparing this paper, we became aware that these
characterizations follow from earlier work on matroids and
Tutte polynomials.  Jaeger's formula~\cite{Jaeger1989}
gives the ternary characterization directly.  The binary
characterization is obtained by combining Vertigan's formula
for Tutte polynomials~\cite{Vertigan1998} with
Greene's relation between Tutte polynomials and weight
enumerators~\cite{Greene1976}.  More explicit expressions for
these special values in the ternary and binary cases were
given by Gioan and Las Vergnas~\cite{GioanLasVergnas2007}
and Pendavingh~\cite{Pendavingh2014}, respectively. 
For completeness, Section~\ref{sec:characterization} presents these
known results using quadratic forms modulo three in the ternary
case and modulo four in the binary case, interpreting the special
values of the weight enumerators as the associated Gauss sums.

Our main contribution in this paper is to derive additional linear constraints on
ordinary weight distributions from these characterizations and to 
use them to strengthen LP bounds for binary and ternary LCD codes. 
For LCD codes, these Gauss sums have only finitely many possible
phases.  Fixing a phase, together with the parity
information, gives linear constraints on the ordinary weight
distributions.  Every LCD code satisfies one of the resulting
LPs, so infeasibility of all of them excludes the prescribed
parameters. To our knowledge, this is the first use of these weight enumerator
characterizations to derive LP bounds for binary and ternary LCD
codes. 

Section~\ref{sec:common} collects the coding-theoretic preliminaries
and MacWilliams identities.  Section~\ref{sec:characterization}
presents the ternary and binary characterizations, phases, and duality
relations, and Section~\ref{sec:gauss-lp} develops the corresponding
LP bounds in the same order.  Section~\ref{sec:computations}
presents the numerical comparisons and exact verification.
Section~\ref{sec:conclusion} summarizes the results.
Appendix~\ref{sec:mixed-comparison} records the
joint-weight-enumerator benchmark.

\section{Preliminaries}
\label{sec:common}

Let $\F_q$ be the finite field of order $q$. An $[n,k,d]_q$ code is a $k$-dimensional linear subspace of $\F_q^n$ with minimum distance $d$.

Throughout, let $C$ be an $[n,k,d]_q$ code with $1\leq k<n$.
For any $\bm x, \bm y \in \F_q^n$, the Euclidean inner product is defined as 
\[ \bm x\cdot \bm y:=\sum_{i=1}^n x_i y_i,\] 
and 
\[C^\perp :=\left\{\bm y \in \F_q^n: \bm x \cdot \bm y=0 \quad \forall \bm x \in C\right\}.\]
The \emph{hull} of $C$ is 
$C\cap C^\perp$; thus $C$ is LCD precisely when its hull is zero.
If $G$ is a generator matrix of $C$ and
$G^{\mathsf T}$ denotes its transpose, the Gram criterion
\cite{Massey1992} states that
\begin{equation}
 C\text{ is LCD}
 \quad\Longleftrightarrow\quad
 GG^{\mathsf T}\text{ is nonsingular}.
\label{eq:gram-criterion}
\end{equation}
In this case, $\F_q^n=C\oplus C^\perp$ with the two summands orthogonal.

For $0\leq j\leq n$, write
\[
 A_j(C)=\abs{\set{\mb c \in C:\wt(\mb c)=j}},
 \qquad
 W_C(z)=\sum_{j=0}^n A_j(C)z^j.
\]
The polynomial $W_C(z)$ is called the weight enumerator of $C$. 
We abbreviate $A_j=A_j(C)$ and $B_j=A_j(C^\perp)$ when the underlying code $C$ is clear from the context. 
To express the relation between $A_j$ and $B_j$, let
\[
 K_j^{(q)}(w)=
 \sum_{\ell=0}^j(-1)^\ell(q-1)^{j-\ell}
 \binom{w}{\ell}\binom{n-w}{j-\ell}, \quad 0 \le j,w \le n
\]
be the $q$-ary Krawtchouk polynomial of length $n$, with binomial coefficients
taken to be zero outside their usual range.  The MacWilliams
identities \cite{MacWilliams1963,MacWilliamsSloane1977} give
\begin{equation}
 q^kB_j=\sum_{w=0}^n A_wK_j^{(q)}(w),
 \qquad 0\leq j\leq n.
\label{eq:macwilliams}
\end{equation}

These equations also make sense for real arrays $(A_j)_{j=0}^n, (B_j)_{j=0}^n$ that need not be
weight distributions.  Their polynomial form will supply the duality
relations used in the LP bounds.

\begin{lemma} \label{lem:formal-macwilliams}
Let real arrays $(A_j)$ and $(B_j)$ satisfy \eqref{eq:macwilliams}, and
put
\[
 A(z)=\sum_{j=0}^nA_jz^j,
 \qquad
 B(z)=\sum_{j=0}^nB_jz^j.
\]
Then
\begin{equation}
 B(z)=q^{-k}\sum_{w=0}^n A_w \left(1+(q-1)z\right)^{n-w} \left(1-z\right)^w.
\label{eq:formal-macwilliams}
\end{equation}
%When the displayed denominator vanishes, the right-hand side is interpreted by polynomial expansion.
\end{lemma}

\begin{proof}
Multiplying the $j$-th equation in \eqref{eq:macwilliams} by $z^j$
and summing over $j$, we obtain \eqref{eq:formal-macwilliams} from
the generating identity for the Krawtchouk polynomials
\cite[Chapter~5]{MacWilliamsSloane1977}:
\[
 \sum_{j=0}^nK_j^{(q)}(w)z^j
 =(1+(q-1)z)^{n-w}(1-z)^w.
\]
\end{proof}

\section{Characterization of ternary and binary LCD codes}
\label{sec:characterization}

We collect the known weight-enumerator characterizations of ternary
and binary LCD codes and the associated phase information.
These follow from earlier work on matroids and Tutte polynomials
\cite{Greene1976,Jaeger1989,Vertigan1998,GioanLasVergnas2007,Pendavingh2014},
while the duality relations follow from the MacWilliams identities.
For completeness, we give detailed proofs using quadratic forms
modulo three and modulo four, in a form suited to the LP bounds
in Section~\ref{sec:gauss-lp}.

\subsection{Ternary LCD codes}
\label{sec:ternary-gauss}

Let $C\leq\F_3^n$ be an $[n,k,d]_3$ code, and put
$Q(\bm x)=\bm x\cdot\bm x$ for $\bm x\in \F_3^n$.  Its associated
symmetric bilinear form is
\[
 \mathfrak{b}(\bm x,\bm y)
 =\frac12\bigl(Q(\bm x+\bm y)-Q(\bm x)-Q(\bm y)\bigr)
 =\bm x\cdot\bm y.
\]
The \emph{radical} of this form is
\[
 R=\set{\bm x\in C:\mathfrak{b}(\bm x,\bm y)=0
       \text{ for all }\bm y\in C}
   =C\cap C^\perp.
\]
We call $Q$ \emph{nondegenerate} if $R=0$; thus $C$ is LCD precisely
when $Q$ is nondegenerate.  Since every nonzero element of $\F_3$
has square $1$,
\[
 \bm x\cdot\bm x\equiv\wt(\bm x)\pmod3,
 \qquad \bm x\in\F_3^n.
\]
Consequently, with $\omega=e^{2\pi\ii/3}$,
\[
 W_C(\omega)
 =\sum_{\bm x\in C}\omega^{\wt(\bm x)}
 =\sum_{\bm x\in C}\omega^{\bm x\cdot\bm x}.
\]
This special value of the weight enumerator is the Gauss sum of the quadratic form $Q$.

\begin{theorem}
\label{thm:ternary-gauss}
With the notation above, put $\rho=\dim R$.  Then
\begin{equation}
 \abs{W_C(\omega)}=3^{(k+\rho)/2}.
\label{eq:ternary-radical}
\end{equation}
Consequently
\[
 C\text{ is LCD}
 \quad\Longleftrightarrow\quad
 \abs{W_C(\omega)}=3^{k/2}.
\]
\end{theorem}

\begin{proof}
Writing $\bm x=\bm y+\bm z$ in the squared magnitude gives
\[
 \abs{W_C(\omega)}^2
 =\sum_{\bm z\in C}\omega^{\bm z\cdot\bm z}
   \sum_{\bm y\in C}\omega^{2\bm y\cdot\bm z}.
\]
For fixed $\bm z$, the map
$\bm y\mapsto\omega^{2\bm y\cdot\bm z}$ is an additive character
of $C$, trivial exactly when $\bm z\in R$.  Character orthogonality
\cite[Chapter~5]{MacWilliamsSloane1977} gives
\[
 \sum_{\bm y\in C}\omega^{2\bm y\cdot\bm z}=
 \begin{cases}
 3^k,&\bm z\in R,\\
 0,&\bm z\notin R.
 \end{cases}
\]
For $\bm z\in R$, we also have $\bm z\cdot\bm z=0$.  Hence
$\abs{W_C(\omega)}^2=3^k3^\rho$, proving
\eqref{eq:ternary-radical}.  The LCD equivalence follows from
\eqref{eq:gram-criterion}.
\end{proof}

To determine the exact value of $W_C(\omega)$ for an LCD code $C$, we need some notation. Let $\chi$ denote the quadratic character of $\F_3^\times$. 
For a generator matrix $G$ of $C$, put
$M_G=GG^{\mathsf T}$ and define
\[
 \varepsilon(C)=\chi(\det M_G)\in\set{1,-1}.
\]
A change of basis of $C$ multiplies $\det M_G$ by a nonzero square, so
$\varepsilon(C)$ is independent of the choice of $G$.

\begin{proposition}
\label{prop:ternary-phase}
If $C\leq\F_3^n$ is a $k$-dimensional LCD code, then
\begin{equation}
 W_C(\omega)=\varepsilon(C)\left(\ii\sqrt3\right)^k.
\label{eq:ternary-phase}
\end{equation}
\end{proposition}

\begin{proof}
Since the restricted Euclidean form is nondegenerate, we may choose
$G$ so that
\[
 M_G=\operatorname{diag}(a_1,\ldots,a_k),
 \qquad a_j\in\F_3^\times.
\]
Then $\varepsilon(C)=\prod_{j=1}^k\chi(a_j)$, and the Gauss sum
factors over the corresponding orthogonal one-dimensional summands.
Using the identity \cite{GioanLasVergnas2007}
\[
 \sum_{t\in\F_3}\omega^{at^2}=\chi(a)\ii\sqrt3,
 \qquad a\in\F_3^\times,
\]
we obtain
\[
 \begin{aligned}
 W_C(\omega)
 &=\sum_{\bm x \in C}\omega^{\bm x \cdot \bm x}=\prod_{j=1}^k\sum_{t\in\F_3}\omega^{a_jt^2}\\
 &=\prod_{j=1}^k\left(\chi(a_j)\ii\sqrt3\right)
  =\varepsilon(C)\left(\ii\sqrt3\right)^k.
 \end{aligned}
\]
\end{proof}

We next evaluate the MacWilliams identity \eqref{eq:formal-macwilliams} at $\omega$.  The resulting
relation holds for arbitrary real arrays $(A_j),(B_j)$ satisfying \eqref{eq:macwilliams}.

\begin{lemma}
\label{prop:formal-ternary}
Let real arrays $(A_j)$ and $(B_j)$ satisfy
\eqref{eq:macwilliams} with $q=3$, and put
\[
 A(z)=\sum_{j=0}^nA_jz^j,
 \qquad
 B(z)=\sum_{j=0}^nB_jz^j.
\]
Then
\begin{equation}
 B(\omega)=3^{-k}\left(\ii\sqrt3\right)^n\overline{A(\omega)}.
\label{eq:formal-ternary-transport}
\end{equation}
Here the overline denotes complex conjugation.
\end{lemma}

\begin{proof}
Apply Lemma~\ref{lem:formal-macwilliams} at $z=\omega$.  Since
\[
 1+2\omega=\ii\sqrt3,
 \qquad
 \frac{1-\omega}{1+2\omega}=\omega^2,
\]
and $A(\omega^2)=\overline{A(\omega)}$ for real coefficients,
\eqref{eq:formal-macwilliams} gives
\eqref{eq:formal-ternary-transport}.
\end{proof}

If $C$ is LCD, then so is $C^\perp$, which has dimension $n-k$.
Proposition~\ref{prop:ternary-phase} and
Lemma~\ref{prop:formal-ternary} therefore give
\[
 \begin{aligned}
 W_{C^\perp}(\omega)
 &=3^{-k}\left(\ii\sqrt3\right)^n\overline{W_C(\omega)}\\
 &=3^{-k}\left(\ii\sqrt3\right)^n
   \varepsilon(C)\left(-\ii\sqrt3\right)^k\\
 &=\varepsilon(C)\left(\ii\sqrt3\right)^{n-k}.
 \end{aligned}
\]
Applying Proposition~\ref{prop:ternary-phase} to $C^\perp$ and
comparing the two expressions yields
\begin{equation}
 \varepsilon(C^\perp)=\varepsilon(C).
\label{eq:ternary-dual-sign}
\end{equation}

\subsection{Binary LCD codes}
\label{sec:binary-gauss}

Let $C\leq\F_2^n$ be an $[n,k,d]_2$ code.  Define
\[
 Q:\F_2^n \longrightarrow \Z/4\Z,
 \qquad Q(\bm x)=\wt(\bm x)\pmod4.
\]
Equivalently, $Q(\bm x)$ is the sum of the squares of the
integer representatives $0,1$ of its coordinates, reduced modulo
four.  For $\bm x,\bm y\in\F_2^n$, we have
\begin{equation}
 Q(\bm x+\bm y)-Q(\bm x)-Q(\bm y)
 \equiv 2\left( \bm x\cdot\bm y\right)\pmod4.
\label{eq:binary-polarization}
\end{equation}
Here vector addition and the inner product are taken over $\F_2$,
and $2(\bm x\cdot\bm y)$ is interpreted as $0$ or $2$ modulo four.
Thus $Q$ is a $\mathbb Z/4\mathbb Z$-valued quadratic form,
or quadratic refinement, associated with the Euclidean bilinear
form $\mathfrak{b}(\bm x,\bm y)=\bm x\cdot\bm y$
\cite{Brown1972,Wood1993}.

The radical of the restricted bilinear form is
\[
 R=C\cap C^\perp.
\]
We call the restriction of $Q$ to $C$ \emph{nondegenerate} if
$R=0$; this is equivalent to $C$ being LCD.  Since
$\ii^{Q(\bm x)}=\ii^{\wt(\bm x)}$ for any $\bm x \in \F_2^n$,
\[
 W_C(\ii)
 =\sum_{\bm x\in C}\ii^{\wt(\bm x)}
 =\sum_{\bm x\in C}\ii^{Q(\bm x)}.
\]
This special value of the weight enumerator is the Gauss sum of the
restricted refinement.

\begin{theorem}
\label{thm:binary-gauss}
With the notation above, put $\rho=\dim R$.  Then
\begin{equation}
 \abs{W_C(\ii)}=
 \begin{cases}
 0,&Q|_R\not\equiv0,\\
 2^{(k+\rho)/2},&Q|_R\equiv0.
 \end{cases}
\label{eq:binary-radical}
\end{equation}
Consequently
\[
 C\text{ is LCD}
 \quad\Longleftrightarrow\quad
 \abs{W_C(\ii)}=2^{k/2}.
\]
\end{theorem}

\begin{proof}
Writing $\bm x=\bm y+\bm z$ and using
\eqref{eq:binary-polarization}, we obtain
\[
 \abs{W_C(\ii)}^2
 =\sum_{\bm z\in C}\ii^{Q(\bm z)}
   \sum_{\bm y\in C}(-1)^{\bm y\cdot\bm z}
 =2^k\sum_{\bm z\in R}\ii^{Q(\bm z)}.
\]
The second equality follows from character orthogonality, as in the
proof of Theorem~\ref{thm:ternary-gauss}.  On $R$,
\eqref{eq:binary-polarization} shows that $Q$ is additive.
Hence $\bm z\mapsto\ii^{Q(\bm z)}$ is a character of $R$,
trivial precisely when $Q|_R\equiv0$.  Its sum is then $2^\rho$,
and is zero otherwise.  This proves \eqref{eq:binary-radical};
the LCD equivalence follows as in the ternary case.
\end{proof}

To describe the phase of $W_C(\ii)$ for an LCD code $C$, put
$\tau(C)=O$ if $C$ contains an odd-weight codeword and $\tau(C)=E$
otherwise.  Since
$Q(\bm x)\equiv\mathfrak{b}(\bm x,\bm x)\pmod2$, the restricted
bilinear form is alternating precisely when $\tau(C)=E$.

\begin{proposition}
\label{prop:binary-phase}
Let $C\leq\F_2^n$ be a $k$-dimensional LCD code.  Then
\begin{equation}
 W_C(\ii)=2^{k/2}\exp\!\left(\frac{\pi\ii\beta(C)}4\right),
\label{eq:binary-phase}
\end{equation}
where $\beta(C)\in\mathbb Z/8\mathbb Z$ is the Brown invariant
of the restricted quadratic refinement.  It satisfies
$\beta(C)\in\mathcal B_{\tau(C)}(k)$, where
\[
 \mathcal B_O(k)=\set{k-2t\pmod8:0\leq t\leq k},
 \qquad
 \mathcal B_E(k)=
 \begin{cases}
 \set{0,4},&k\text{ even},\\
 \varnothing,&k\text{ odd}.
 \end{cases}
\]
\end{proposition}

\begin{proof}
Since $C$ is LCD, its restricted bilinear form is nondegenerate.
By \eqref{eq:binary-polarization},
$Q(\bm x+\bm y)=Q(\bm x)+Q(\bm y)$ whenever
$\bm x\cdot\bm y=0$.  Consequently, the Gauss sum factors over
orthogonal direct sums.

Suppose first that $\tau(C)=O$.  The form is nonalternating,
so $C$ admits an orthonormal basis
$\bm e_1,\ldots,\bm e_k$
\cite[Theorem~3.1]{CarletEtAl2019}.
Since
$Q(\bm e_j)\equiv\mathfrak{b}(\bm e_j,\bm e_j)=1\pmod2$,
each $Q(\bm e_j)$ is $1$ or $3$ modulo four.
If exactly $t$ of these values are $3$, then
\[
 \begin{aligned}
 W_C(\ii)
 &=\prod_{j=1}^k\bigl(1+\ii^{Q(\bm e_j)}\bigr)\\
 &=(1+\ii)^{k-t}(1-\ii)^t
 =2^{k/2}\exp\!\left(\frac{\pi\ii(k-2t)}4\right).
 \end{aligned}
\]
Thus $\beta(C)\equiv k-2t\pmod8$.

Suppose now that $\tau(C)=E$.  The form is alternating, so $k$
is even and $C$ is an orthogonal direct sum of $k/2$ symplectic
planes \cite[Theorem~3.2]{CarletEtAl2019}.
On such a plane $P$ with symplectic basis $\bm e,\bm f$,
both $Q(\bm e)$ and $Q(\bm f)$ are even.  Moreover,
$\mathfrak{b}(\bm e,\bm f)=1$, so
\eqref{eq:binary-polarization} gives
\[
 Q(\bm e+\bm f)
 \equiv Q(\bm e)+Q(\bm f)+2\pmod4.
\]
Hence the Gauss sum on $P$ is
\[
 \sum_{\bm x\in P}\ii^{Q(\bm x)}
 =1+\ii^{Q(\bm e)}+\ii^{Q(\bm f)}
  -\ii^{Q(\bm e)+Q(\bm f)}
 \in\set{2,-2}.
\]
Multiplying over the $k/2$ planes gives
$W_C(\ii)\in\set{2^{k/2},-2^{k/2}}$, and therefore
$\beta(C)\in\set{0,4}$.

These two cases establish \eqref{eq:binary-phase} and the stated
phase restrictions.  The phase index defined by this Gauss-sum
identity is the Brown invariant
\cite[Section~2]{Wood1993}.
\end{proof}

We next evaluate the MacWilliams identity of
Lemma~\ref{lem:formal-macwilliams} at $\ii$ and $-1$.
The resulting relations hold for arbitrary real arrays satisfying
\eqref{eq:macwilliams}.

\begin{lemma}
\label{prop:formal-binary}
Let real arrays $(A_j)$ and $(B_j)$ satisfy
\eqref{eq:macwilliams} with $q=2$, and put
\[
 A(z)=\sum_{j=0}^nA_jz^j,
 \qquad
 B(z)=\sum_{j=0}^nB_jz^j.
\]
Then
\begin{align}
 B(\ii)&=2^{-k}(1+\ii)^n\overline{A(\ii)},
\label{eq:formal-binary-transport}\\
 B(-1)&=2^{n-k}A_n.
\label{eq:dual-minus-one}
\end{align}
\end{lemma}

\begin{proof}
Apply Lemma~\ref{lem:formal-macwilliams} at $z=\ii$.  Since
$(1-\ii)/(1+\ii)=-\ii$ and
$A(-\ii)=\overline{A(\ii)}$ for real coefficients, we obtain
\eqref{eq:formal-binary-transport}.  At $z=-1$, only the $j=n$
term survives in the expanded MacWilliams polynomial, giving
\eqref{eq:dual-minus-one}.
\end{proof}

If $C$ is LCD, then so is $C^\perp$, which has dimension $n-k$.
Proposition~\ref{prop:binary-phase} and
Lemma~\ref{prop:formal-binary} therefore give
\[
 \begin{aligned}
 W_{C^\perp}(\ii)
 &=2^{-k}(1+\ii)^n\overline{W_C(\ii)}\\
 &=2^{(n-k)/2}\exp\!\left(\frac{\pi\ii(n-\beta(C))}{4}\right).
 \end{aligned}
\]
Comparing with Proposition~\ref{prop:binary-phase} applied to
$C^\perp$ yields
\begin{equation}
 \beta(C^\perp)\equiv n-\beta(C)\pmod8.
\label{eq:binary-dual-brown}
\end{equation}

The value $W_{C^\perp}(-1)$ counts the even-weight codewords of
$C^\perp$ minus its odd-weight codewords.  Since
$\abs{C^\perp}=2^{n-k}$, \eqref{eq:dual-minus-one} gives
\begin{equation}
 \begin{aligned}
 \tau(C^\perp)=E
 &\Longleftrightarrow W_{C^\perp}(-1)=2^{n-k}\\
 &\Longleftrightarrow A_n(C)=1\\
 &\Longleftrightarrow \bm 1\in C,
 \end{aligned}
\label{eq:binary-dual-type}
\end{equation}
where $\bm 1=(1,\ldots,1)$ is the unique binary vector of weight $n$.
The pair $(E,E)$ is impossible for an LCD code, since
\eqref{eq:binary-dual-type}, applied to both $C$ and $C^\perp$,
would give $\bm 1\in C\cap C^\perp$.

\section{Linear programming bounds}
\label{sec:gauss-lp}

Section~\ref{sec:characterization} gives finitely many possible
Gauss phases for LCD codes.  Once a phase and, in the binary case,
a parity type are fixed, the total number of codewords in each weight residue class is determined.  We use these
constraints to strengthen the LCD counting LP
while retaining only the ordinary weight distributions.

Throughout this section, fix
\[
 q\in\set{2,3},\qquad 1\leq k<n,\qquad 1\leq d\leq n-k+1.
\]
Here $d$ is a proposed lower bound on the minimum distance; the upper
limit of $d$ is the Singleton bound.  The Delsarte LP
\cite{Delsarte1973} uses real variables
$A_0,\ldots,A_n,B_0,\ldots,B_n$ satisfying
\begingroup
\let\baselineLabel\label
\begin{equation}
 \begin{aligned}
 A_0&=B_0=1,\\
 A_w&\geq0,\quad B_w\geq0 &&(0\leq w\leq n),\\
 A_w&=0 &&(1\leq w<d),\\
 q^kB_j&=\sum_{w=0}^n A_wK_j^{(q)}(w) &&(0\leq j\leq n).
 \end{aligned}
\label{eq:hamming-lp}
% Preserve labels from the former alphabet-specific baseline displays.
\baselineLabel{eq:g3-norm}
\baselineLabel{eq:g3-nonneg}
\baselineLabel{eq:g3-support}
\baselineLabel{eq:g3-mac}
\baselineLabel{eq:g2-norm}
\baselineLabel{eq:g2-nonneg}
\baselineLabel{eq:g2-support}
\baselineLabel{eq:g2-mac}
\end{equation}
\endgroup
The last equations are the MacWilliams identities
\eqref{eq:macwilliams}.  Their $j=0$ case and
Lemma~\ref{lem:formal-macwilliams} at $z=1$ give, respectively,
\[
 \sum_{w=0}^n A_w=q^k,\qquad
 \sum_{w=0}^n B_w=q^{n-k}.
\]
These normalizations hold for the real LP variables, whether or not
they arise from a code.

For an LCD code $C$, the weight-$w$ codewords in $C$ and $C^\perp$ are
distinct whenever $w>0$.  Their combined number cannot exceed the
number of ambient vectors of weight $w$, so we also impose
\begingroup
\let\baselineLabel\label
\begin{equation}
 A_w+B_w\leq(q-1)^w\binom nw,
 \qquad 1\leq w\leq n.
\label{eq:lcd-shell}
\baselineLabel{eq:g3-shell}
\baselineLabel{eq:g2-shell}
\end{equation}
\endgroup
The counting inequalities \eqref{eq:lcd-shell} were used by
Dougherty et al.\ \cite{DoughertyEtAl2017} for binary LCD codes
and extended to ternary codes by Karabakla and
\"Ozkaya~\cite{KarabaklaOzkaya2026}.  We use their counting
constraints in the following real-valued formulation at fixed
dimension.

\begin{definition}[LCD counting LP]
\label{def:lcd-counting-lp}
For fixed $q\in\set{2,3}$, $1\leq k<n$, and
$1\leq d\leq n-k+1$, let $\mathcal H_q(n,k,d)$ denote the linear
feasibility problem over real variables $(A_w)_{w=0}^n$ and
$(B_w)_{w=0}^n$ defined by \eqref{eq:hamming-lp} and
\eqref{eq:lcd-shell}.  We call this the \emph{LCD counting LP}.
For fixed $(n,k)$, its distance bound is
\[
 D_{\mathrm H}^{(q)}(n,k)
 =\max\set{d\in\set{1,\ldots,n-k+1}:
                \mathcal H_q(n,k,d)\text{ is feasible}}.
\]
\end{definition}

Every $q$-ary $[n,k,d]$ LCD code satisfies
$d\leq D_{\mathrm H}^{(q)}(n,k)$, since the weight distributions
of the code and its dual satisfy all the defining constraints.
The comparisons below use precisely this relaxation over real
variables. Each refinement adds residue equations for
a fixed choice of phase and, in the binary case, parity data.
We call each such choice a \emph{branch}.

\subsection{Ternary LP bounds}
\label{subsec:ternary-lp}

Let $C\leq\F_3^n$ be a $k$-dimensional LCD code.
For $a\in\set{0,1,2}$, define 
\[
 R_a(C)=
 \sum_{\substack{0\leq j\leq n\\j\equiv a \pmod{3}}} A_j(C).
\]
Using $1+\omega+\omega^2=0$, with $\omega=e^{2\pi\ii/3}$, we obtain
\[
 R_a(C)
 =\frac13\sum_{t=0}^2\omega^{-at}W_C(\omega^t).
\]
Since $W_C(1)=3^k$ and
$W_C(\omega^2)=\overline{W_C(\omega)}$,
Proposition~\ref{prop:ternary-phase} yields
\begin{equation}
 R_a(C)=
 \frac13\left[3^k+2\Real\!\left(
 \omega^{-a}\varepsilon(C)(\ii\sqrt3)^k\right)\right],
 \qquad a=0,1,2.
\label{eq:ternary-residues}
\end{equation}
Thus $\varepsilon(C)$ completely determines the residue totals
$R_a(C)$.  Each identity in \eqref{eq:ternary-residues} gives a
linear constraint on the $A_j$.  Since the totals sum to $3^k$,
we retain only the $a=0,1$ equations.

For $\varepsilon(C)=\epsilon\in\set{1,-1}$, these two totals $R_0(C), R_1(C)$ can be evaluated: 
\begin{equation}
 \left(R_0(C),R_1(C)\right)=
 \begin{cases}
 \left(3^{k-1}-2\epsilon(-3)^{(k-2)/2},
       3^{k-1}+\epsilon(-3)^{(k-2)/2}\right),&k\text{ even},\\
 \left(3^{k-1},
       3^{k-1}+\epsilon(-3)^{(k-1)/2}\right),&k\text{ odd}.
 \end{cases}
\label{eq:ternary-gauss-lp}
\end{equation}
Denote the two entries on the right by $T_0(\epsilon),T_1(\epsilon)$, respectively.

\begin{definition}
\label{def:ternary-gauss-lp}
For fixed $(n,k,d)$ and $\epsilon\in\set{1,-1}$, the
$\epsilon$-branch consists of the LCD counting LP
$\mathcal H_3(n,k,d)$ from Definition~\ref{def:lcd-counting-lp},
together with
\begin{align}
 \sum_{\substack{0\leq j\leq n\\j\equiv0 \pmod{3}}}A_j&=T_0(\epsilon),
\label{eq:g3-res0}\\
 \sum_{\substack{0\leq j\leq n\\j\equiv1 \pmod{3}}}A_j&=T_1(\epsilon).
\label{eq:g3-res1}
\end{align}
\end{definition}

Through the MacWilliams identities \eqref{eq:macwilliams}, the residue equations for
$(B_j)$ also give linear constraints on $(A_j)$.  By
Lemma~\ref{prop:formal-ternary}, however, these constraints
already follow from the two residue equations above and are
therefore redundant.

We call the ternary LP relaxation feasible if at least one of
the two branches $\epsilon\in\set{1,-1}$ in
Definition~\ref{def:ternary-gauss-lp} is feasible.

\begin{proposition}
\label{prop:ternary-lp-bound}
For fixed $(n,k)$, let $D$ be the largest integer
$1\leq d\leq n-k+1$ for which the ternary LP relaxation is
feasible.  Then every ternary $[n,k,d]$ LCD code satisfies
$d\leq D$.
\end{proposition}

\begin{proof}
Let $C$ be a ternary $[n,k,d]$ LCD code.  The weight
distributions of $C$ and $C^\perp$ satisfy the branch with
$\epsilon=\varepsilon(C)$.  Thus the relaxation is feasible
at $d$, and hence $d\leq D$.
\end{proof}

We call $D$ the ternary Gauss-phase LP distance bound.
It can be computed by testing $d=n-k+1,n-k,n-k-1,\ldots$ in
descending order until at least one branch is feasible.
Since the support constraints weaken as $d$ decreases,
the first feasible value is $D$.
\subsection{Binary LP bounds}
\label{subsec:binary-lp}

Let $C\leq\F_2^n$ be a $k$-dimensional LCD code.
For $a\in\set{0,1,2,3}$, define
\[
 R_a(C)=
 \sum_{\substack{0\leq j\leq n\\j\equiv a\pmod{4}}}A_j(C).
\]
As in the ternary case, we have
\[
 R_a(C)=\frac14\sum_{t=0}^3\ii^{-at}W_C(\ii^t).
\]
Since $W_C(1)=2^k$, $W_C(-1)=\Delta_{\tau(C)}(k)$, where
\[
 \Delta_\tau(k)=
 \begin{cases}
 0,&\tau=O,\\
 2^k,&\tau=E,
 \end{cases}
\] 
and 
$W_C(-\ii)=\overline{W_C(\ii)}$, 
Proposition~\ref{prop:binary-phase} yields
\begin{equation}
 R_a(C)=
 \frac14\left[
 2^k+(-1)^a\Delta_{\tau(C)}(k)
 +2\Real\!\left(
 \ii^{-a}2^{k/2}e^{\pi\ii\beta(C)/4}
 \right)\right],
 \quad a=0,1,2,3.
\label{eq:binary-residues}
\end{equation}
Thus $\tau(C)$ and $\beta(C)$ completely determine the residue
totals $R_a(C)$.  Each identity in \eqref{eq:binary-residues} gives
a linear constraint on the $A_j$.  Since the totals sum to $2^k$,
we retain only the $a=0,1,2$ equations.

For $\tau(C)=\tau \in \{O,E\}$ and $\beta(C)=\beta\in\mathcal B_\tau(k)$,
the three retained totals are
\[
 \begin{aligned}
 R_0(C)&=\frac14\left(2^k+\Delta_\tau(k)
              +2^{k/2+1}\cos\frac{\pi\beta}{4}\right),\\
 R_1(C)&=\frac14\left(2^k-\Delta_\tau(k)
              +2^{k/2+1}\sin\frac{\pi\beta}{4}\right),\\
 R_2(C)&=\frac14\left(2^k+\Delta_\tau(k)
              -2^{k/2+1}\cos\frac{\pi\beta}{4}\right).
 \end{aligned}
\]
Denote the three expressions on the right by
$T_0(\tau,\beta),T_1(\tau,\beta),T_2(\tau,\beta)$, respectively.
The factors $2^{k/2}\cos(\pi\beta/4)$ and
$2^{k/2}\sin(\pi\beta/4)$ are integers, since
$\beta\equiv k\pmod2$.

The branch choice must also account for the parity type of
$C^\perp$.  Since the all-one vector is the unique binary vector
of weight $n$, $A_n(C)\in\set{0,1}$ records whether it lies in $C$.
By \eqref{eq:binary-dual-type}, this determines $\tau^\perp:=\tau(C^\perp)$.
In the LP, we impose $A_n=\eta$ with $\eta\in\set{0,1}$ and set
\[
 \tau^\perp=
 \begin{cases}
 E,&\eta=1,\\
 O,&\eta=0.
 \end{cases}
\]
We call $\eta$ the \emph{endpoint selector}.
By \eqref{eq:binary-dual-brown}, the dual phase is
$\beta^\perp :=\beta(C^\perp) \equiv n-\beta\pmod8$; moreover, the pair $(E,E)$ is impossible.
Accordingly, a triple $(\tau,\beta,\eta)$ with
$\tau\in\set{O,E}$, $\beta\in\mathcal B_\tau(k)$, and
$\eta\in\set{0,1}$ is called \emph{admissible} if
\[
 (\tau,\tau^\perp)\ne(E,E),
 \qquad
 (n-\beta)\bmod8\in\mathcal B_{\tau^\perp}(n-k).
\]
There are at most eight admissible branches: the type pairs
$(O,O)$, $(O,E)$, and $(E,O)$ allow at most four, two, and two
phases, respectively.

\begin{definition}
\label{def:binary-gauss-lp}
For fixed $(n,k,d)$ and an admissible triple $(\tau,\beta,\eta)$,
the $(\tau,\beta,\eta)$-branch consists of the LCD counting LP
$\mathcal H_2(n,k,d)$ from Definition~\ref{def:lcd-counting-lp},
together with
\begin{align}
 \sum_{\substack{0\leq j\leq n\\j\equiv0\pmod{4}}}A_j
 &=T_0(\tau,\beta),
\label{eq:g2-res0}\\
 \sum_{\substack{0\leq j\leq n\\j\equiv1\pmod{4}}}A_j
 &=T_1(\tau,\beta),
\label{eq:g2-res1}\\
 \sum_{\substack{0\leq j\leq n\\j\equiv2\pmod{4}}}A_j
 &=T_2(\tau,\beta),
\label{eq:g2-res2}\\
 A_n&=\eta.
\label{eq:binary-endpoint-selector}
\end{align}
\end{definition}

Through the MacWilliams identities \eqref{eq:macwilliams}, the
residue equations for $(B_j)$ also give linear constraints on
$(A_j)$.  By Lemma~\ref{prop:formal-binary}, however, these
constraints follow from the three residue equations above and
$A_n=\eta$, and are therefore redundant.

We call the binary LP relaxation feasible if at least one of
the admissible branches in Definition~\ref{def:binary-gauss-lp}
is feasible.

\begin{proposition}
\label{prop:binary-lp-bound}
For fixed $(n,k)$, let $D$ be the largest integer
$1\leq d\leq n-k+1$ for which the binary LP relaxation is
feasible.  Then every binary $[n,k,d]$ LCD code satisfies
$d\leq D$.
\end{proposition}

\begin{proof}
Let $C$ be a binary $[n,k,d]$ LCD code.  The weight distributions
of $C$ and $C^\perp$ satisfy the admissible branch
$(\tau(C),\beta(C),A_n(C))$.  Thus the relaxation is feasible
at $d$, and hence $d\leq D$.
\end{proof}

We call $D$ the binary Gauss-phase LP distance bound.
As in the ternary case, it can be computed by testing
$d=n-k+1,n-k,n-k-1,\ldots$ in descending order until at least
one admissible branch is feasible.

Both refinements retain the $2(n+1)$ real variables of the LCD
counting LP, adding two residue equations in the ternary case and
three residue equations together with $A_n=\eta$ in the binary case.
Since every branch retains the baseline constraints, the resulting
distance bounds cannot exceed the baseline bounds.  Feasibility is
only a necessary condition for the existence of a code.

\section{Numerical bounds and exact verification}
\label{sec:computations}
We compare the Gauss-phase bounds with the LCD counting LP
\cite{DoughertyEtAl2017,KarabaklaOzkaya2026} and, in the binary case,
the joint-weight-enumerator LP \cite{AlahmadiEtAl2019}.

\subsection{Bounds and parameter ranges}
\label{subsec:numerical-protocol}
Write $D_{\mathrm H}^{(q)}(n,k)$ and $D_{\mathrm G}^{(q)}(n,k)$ for
the LCD counting and Gauss-phase bounds defined in
Section~\ref{sec:gauss-lp}. We suppress $q,n,k$ when clear.
Our comparisons concern these real-valued LP relaxations, rather than
published tables incorporating other upper bounds. A feasible LP point
need not correspond to a code.

We consider all $1\leq k<n$ for binary lengths $2\leq n\leq20$
and ternary lengths $2\leq n\leq14$, giving $190$ and $91$
parameter pairs, respectively.

Exactly verified bounds are supported by rational feasible points and
integer infeasibility certificates; values obtained only through
floating-point optimization are called screening values. Verification
details are provided in Online Resource~1.

\subsection{Comparison with the LCD counting LP}
\label{subsec:hamming-gauss-comparison}
Table~\ref{tab:scan-summary} summarizes the comparison with the LCD counting LP.
The Gauss-phase bound is strictly smaller for $62$ binary and $39$ ternary
pairs, listed in Tables~\ref{tab:all-gains-2} and~\ref{tab:all-gains-3};
both endpoints in all $101$ comparisons are verified exactly.
The remaining $128$ binary and $52$ ternary pairs have equal screening values.
Every improvement is by one except at the binary pair $(15,4)$,
where the bound decreases from $8$ to $6$.

\begin{table}[htbp]
\centering
\small
\caption{Comparison with the LCD counting LP.}
\label{tab:scan-summary}
\begin{tabular}{clrrrc}
\toprule
$q$ & Length range & Pairs & Equal & Strict & Largest decrease \\
\midrule
2 & $2\leq n\leq 10$ & 45 & 38 & 7 & 1 \\
2 & $11\leq n\leq 15$ & 60 & 33 & 27 & 2 \\
2 & $16\leq n\leq 20$ & 85 & 57 & 28 & 1 \\
3 & $2\leq n\leq 8$ & 28 & 20 & 8 & 1 \\
3 & $9\leq n\leq 11$ & 27 & 13 & 14 & 1 \\
3 & $12\leq n\leq 14$ & 36 & 19 & 17 & 1 \\
\bottomrule
\end{tabular}
\end{table}

\begin{table}[htbp]
\centering
\small
\setlength{\tabcolsep}{5pt}
\caption{Binary strict improvements over the LCD counting LP.}
\label{tab:all-gains-2}
\begin{tabular}{rrcc@{\quad}|@{\quad}rrcc@{\quad}|@{\quad}rrcc}
\toprule
$n$ & $k$ & $D_{\mathrm H}$ & $D_{\mathrm G}$ & $n$ & $k$ & $D_{\mathrm H}$ & $D_{\mathrm G}$ & $n$ & $k$ & $D_{\mathrm H}$ & $D_{\mathrm G}$ \\
\midrule
5 & 2 & 3 & 2 & 14 & 4 & 7 & 6 & 18 & 4 & 9 & 8 \\
6 & 2 & 4 & 3 & 14 & 5 & 6 & 5 & 18 & 5 & 8 & 7 \\
6 & 3 & 3 & 2 & 14 & 6 & 6 & 5 & 18 & 6 & 8 & 7 \\
7 & 3 & 4 & 3 & 14 & 7 & 5 & 4 & 18 & 7 & 7 & 6 \\
7 & 4 & 3 & 2 & 14 & 9 & 4 & 3 & 18 & 9 & 6 & 5 \\
8 & 3 & 4 & 3 & 14 & 10 & 3 & 2 & 18 & 10 & 5 & 4 \\
10 & 5 & 4 & 3 & 15 & 3 & 8 & 7 & 19 & 4 & 10 & 9 \\
11 & 2 & 7 & 6 & 15 & 4 & 8 & 6 & 19 & 5 & 9 & 8 \\
11 & 3 & 6 & 5 & 15 & 5 & 7 & 6 & 19 & 7 & 8 & 7 \\
11 & 4 & 5 & 4 & 15 & 7 & 6 & 5 & 19 & 8 & 7 & 6 \\
12 & 2 & 8 & 7 & 15 & 8 & 5 & 4 & 19 & 10 & 6 & 5 \\
12 & 4 & 6 & 5 & 15 & 10 & 4 & 3 & 19 & 11 & 5 & 4 \\
12 & 5 & 5 & 4 & 15 & 11 & 3 & 2 & 20 & 3 & 11 & 10 \\
12 & 7 & 4 & 3 & 16 & 4 & 8 & 7 & 20 & 5 & 10 & 9 \\
12 & 8 & 3 & 2 & 16 & 5 & 7 & 6 & 20 & 6 & 9 & 8 \\
13 & 3 & 7 & 6 & 16 & 7 & 6 & 5 & 20 & 7 & 8 & 7 \\
13 & 5 & 6 & 5 & 17 & 2 & 11 & 10 & 20 & 8 & 8 & 7 \\
13 & 6 & 5 & 4 & 17 & 5 & 8 & 7 & 20 & 9 & 7 & 6 \\
13 & 8 & 4 & 3 & 17 & 6 & 7 & 6 & 20 & 11 & 6 & 5 \\
13 & 9 & 3 & 2 & 18 & 2 & 12 & 11 & 20 & 12 & 5 & 4 \\
14 & 3 & 8 & 7 & 18 & 3 & 10 & 9 &  &  &  &  \\
\bottomrule
\end{tabular}
\end{table}

\begin{table}[htbp]
\centering
\small
\setlength{\tabcolsep}{5pt}
\caption{Ternary strict improvements over the LCD counting LP.}
\label{tab:all-gains-3}
\begin{tabular}{rrcc@{\quad}|@{\quad}rrcc@{\quad}|@{\quad}rrcc}
\toprule
$n$ & $k$ & $D_{\mathrm H}$ & $D_{\mathrm G}$ & $n$ & $k$ & $D_{\mathrm H}$ & $D_{\mathrm G}$ & $n$ & $k$ & $D_{\mathrm H}$ & $D_{\mathrm G}$ \\
\midrule
3 & 1 & 3 & 2 & 10 & 4 & 6 & 5 & 12 & 6 & 6 & 5 \\
3 & 2 & 2 & 1 & 10 & 5 & 5 & 4 & 12 & 8 & 4 & 3 \\
4 & 2 & 3 & 2 & 10 & 6 & 4 & 3 & 12 & 9 & 3 & 2 \\
6 & 1 & 6 & 5 & 11 & 2 & 8 & 7 & 12 & 11 & 2 & 1 \\
6 & 5 & 2 & 1 & 11 & 3 & 7 & 6 & 13 & 3 & 9 & 8 \\
7 & 2 & 5 & 4 & 11 & 5 & 6 & 5 & 13 & 4 & 8 & 7 \\
8 & 2 & 6 & 5 & 11 & 6 & 5 & 4 & 13 & 5 & 7 & 6 \\
8 & 3 & 5 & 4 & 11 & 7 & 4 & 3 & 13 & 9 & 4 & 3 \\
9 & 1 & 9 & 8 & 11 & 8 & 3 & 2 & 13 & 10 & 3 & 2 \\
9 & 3 & 6 & 5 & 12 & 1 & 12 & 11 & 14 & 3 & 9 & 8 \\
9 & 4 & 5 & 4 & 12 & 2 & 9 & 8 & 14 & 4 & 9 & 8 \\
9 & 5 & 4 & 3 & 12 & 3 & 8 & 7 & 14 & 6 & 7 & 6 \\
9 & 8 & 2 & 1 & 12 & 4 & 7 & 6 & 14 & 10 & 4 & 3 \\
\bottomrule
\end{tabular}
\end{table}
\FloatBarrier

\begin{remark}[Effect of the endpoint selector]
\label{rem:endpoint-selector}
At $(n,k)=(29,20)$, compatibility forces $A_{29}=0$.
The Gauss-phase endpoint is exactly $4$ with this selector and $5$ without it.
The selector-free endpoint follows from a rational witness at distance $5$
and a counting-LP exclusion at distance $6$; complete verification
is in Online Resource~1.
\end{remark}

\subsection{Comparison with the binary joint-weight-enumerator LP}
\label{subsec:mixed-numerics}
The binary joint-weight-enumerator LP, defined in
Appendix~\ref{sec:mixed-comparison}, uses joint coefficients to retain
information about a code and its dual. Denote its distance bound by
$D_{\mathrm M}(n,k)$. We compare it with the Gauss-phase bound on the
$62$ binary pairs in
Table~\ref{tab:all-gains-2}, where the latter already improves the LCD
counting bound.

The two bounds have equal screening values in $57$ cases.
Table~\ref{tab:gauss-mixed-exact} lists the five differences: the
Gauss-phase bound is stronger in four cases and the joint-weight-enumerator
bound is stronger in one. The four Gauss-phase improvements are verified
exactly. At $(20,8)$, the Gauss-phase value $7$ is exact, whereas the
joint-weight-enumerator value $6$ remains a screening result. Thus these
computations do not imply that either LP dominates the other in general.

\begin{table}[htbp]
\centering
\small
\caption{Differences between the Gauss-phase and joint-LP bounds.}
\label{tab:gauss-mixed-exact}
\begin{tabular}{rrcc}
\toprule
$n$ & $k$ & $D_{\mathrm M}$ & $D_{\mathrm G}$ \\
\midrule
14 & 5 & 6 & 5 \\
16 & 7 & 6 & 5 \\
18 & 5 & 8 & 7 \\
20 & 7 & 8 & 7 \\
20 & 8 & 6 & 7 \\
\bottomrule
\end{tabular}
\end{table}
\FloatBarrier

\section{Conclusion}
\label{sec:conclusion}
We have converted Gauss-phase information in binary and ternary weight enumerators into linear constraints for bounding the minimum distance of LCD codes. The resulting method preserves the $2(n+1)$ ordinary weight-distribution variables and requires only a bounded number of branches and additional equations. Its contribution is the use of the known special-value structure in a compact LP strengthening.

The complete scans through lengths $20$ and $14$, respectively, give $62$ binary and $39$ ternary strict improvements over the real-valued LCD counting LP, with exact endpoint verification for all $101$ comparisons. Four binary comparisons also give exact improvements over the joint-weight-enumerator LP. These results concern the specified relaxations and do not establish dominance over the joint model for all parameters or improvements over tables that combine several upper bounds.

One direction for further work is to impose the Gauss-phase equations within the joint-weight-enumerator model and determine when the two kinds of constraints yield complementary improvements. Another is to identify additional arithmetic constraints that strengthen the ordinary-weight formulation while preserving its small size.

\appendix
\section{The binary joint-weight-enumerator benchmark}
\label{sec:mixed-comparison}

For binary LCD codes we compare the Gauss-phase LP with the
joint-weight-enumerator LP of Alahmadi, Deza, Dutour Sikiri\'c, and
Sol\'e \cite{AlahmadiEtAl2019}.  This model is used only as a benchmark,
so we record the constraints needed for the comparison and refer to
\cite{AlahmadiEtAl2019} for the derivation of the joint MacWilliams
identity.

Put $D=C^\perp$.  For $(\bm x,\bm y)\in C\times D$, let $(a,b,c,e)$ count the
coordinate states $00,01,10,11$, respectively, and put
\[
 \mathcal I_n=
 \set{(a,b,c,e)\in\mathbb Z_{\geq0}^4:a+b+c+e=n}.
\]
For an actual code, $M_{a,b,c,e}$ counts the corresponding ordered
pairs $(x,y)$.  These four counts determine
\[
 \wt(\bm x)=c+e,\qquad
 \wt(\bm y)=b+e,\qquad
 \wt(\bm x+\bm y)=b+c.
\]
Since $C$ is LCD, $C\oplus D=\F_2^n$, so $(\bm x,\bm y)\mapsto \bm x+\bm y$ is a
bijection from $C\times D$ onto $\F_2^n$.  This explains the code,
dual, and ambient-shell marginals below.

\begin{definition}[Binary joint-weight-enumerator LP]
\label{def:mixed-benchmark}
Fix $1\leq k<n$ and $1\leq d\leq n-k+1$.  The benchmark
$\mathcal M_2(n,k,d)$ is the linear feasibility problem in
\[
 A_w,B_w\quad(0\leq w\leq n),
 \qquad
 M_{a,b,c,e}\quad((a,b,c,e)\in\mathcal I_n).
\]
The variables $A_w,B_w$ satisfy the LCD counting LP
$\mathcal H_2(n,k,d)$ of Definition~\ref{def:lcd-counting-lp}.
In addition,
\begin{align}
 M_{a,b,c,e}&\geq0,
 &&(a,b,c,e)\in\mathcal I_n,
\label{eq:mixed-nonneg}\\
 M_{a,b,c,e}&=0,
 &&e\text{ odd},
\label{eq:mixed-orth}\\
 M_{a,b,c,e}&=0,
 &&1\leq c+e<d,
\label{eq:mixed-support}\\
 M_{a,0,0,e}&=0,
 &&e>0.
\label{eq:mixed-diagonal}
\end{align}
Here \eqref{eq:mixed-orth} is $x\cdot y=0$, while
\eqref{eq:mixed-diagonal} is $C\cap D=\set{0}$.  For
$0\leq w\leq n$ impose
\begin{align}
 M_{n-w,w,0,0}&=B_w,&
 M_{n-w,0,w,0}&=A_w,
\label{eq:mixed-boundary}\\
 \sum_{c+e=w}M_{a,b,c,e}&=2^{n-k}A_w,&
 \sum_{b+e=w}M_{a,b,c,e}&=2^kB_w,
\label{eq:mixed-marginals}
\end{align}
where each sum is over $(a,b,c,e)\in\mathcal I_n$, and for
$0\leq j\leq n$ impose
\begin{equation}
 \sum_{b+c=j}M_{a,b,c,e}=\binom nj,
\label{eq:mixed-shell-sums}
\end{equation}
again over $(a,b,c,e)\in\mathcal I_n$.

Finally, define the joint enumerator
\[
 J_{CD}(x_0,x_1,x_2,x_3)=
 \sum_{(a,b,c,e)\in\mathcal I_n}
 M_{a,b,c,e}x_0^a x_1^b x_2^c x_3^e
\]
and
\[
 H_4=
 \begin{pmatrix}
 1&1&1&1\\
 1&1&-1&-1\\
 1&-1&1&-1\\
 1&-1&-1&1
 \end{pmatrix}.
\]
Writing $\mb x=(x_0,x_1,x_2,x_3)^{\mathsf T}$, impose the
four-variable MacWilliams invariance
\begin{equation}
 J_{CD}(\mb x)=2^{-n}J_{CD}(H_4\mb x).
\label{eq:mixed-invariance}
\end{equation}
\end{definition}

Equation~\eqref{eq:mixed-invariance} is the joint MacWilliams identity
from \cite{AlahmadiEtAl2019}.  Equating coefficients gives an ordinary
linear system.  Online Resource~1 contains code to reconstruct the coefficient
matrix and to check the rational witnesses against all these constraints.  For fixed
$(n,k)$, the joint-weight-enumerator distance bound is the largest $d$ for which
$\mathcal M_2(n,k,d)$ is feasible.

\FloatBarrier
\section*{Data and code availability}

The frozen companion archive, Online Resource~1 (version 2026-09-13),
contains the code, exact certificates, rational feasibility witnesses,
and machine-readable results supporting Section~\ref{sec:computations},
together with the computational supplement \emph{Exact verification and
independent checks}. Its \texttt{README.md} gives the artifact map,
software requirements, and reproduction commands. The archive is
self-contained and does not require access to the
\href{https://github.com/kmsming-prog/gauss-phase-lp-lcd-codes}{development
repository}%, which is private during peer review. 
%The frozen archive is the reference version for the computations reported here.

\section*{Statements and Declarations}

\subsection*{Funding}

The work of Ming-Hsuan Kang was supported by the National Science and
Technology Council, Taiwan, under Grant NSTC 115-2115-M-A49-010. The work of Maosheng Xiong was supported by the Research Grant Council (RGC) of Hong Kong under Grant no. 16307524. 

\subsection*{Competing interests}

The authors declare no competing interests.

\subsection*{Use of generative AI}

During manuscript preparation, the authors used OpenAI Codex to assist
with manuscript restructuring, software integration, and computational
consistency checks.  The authors independently reviewed and validated
the mathematical statements, computations, and final text and take full
responsibility for the work.

\end{document}